\documentclass[reqno]{amsart}
\usepackage{amssymb}

\newtheorem{theorem}{Theorem}[section]
\newtheorem{proposition}[theorem]{Proposition}

\theoremstyle{remark}
\newtheorem{remark}[theorem]{Remark}

\numberwithin{equation}{section}

\begin{document}

\title[semi-infinite $q$-boson system]
{The semi-infinite $q$-boson system\\ with boundary interaction}

\author{J.F.  van Diejen}

\author{E. Emsiz}

\address{
Facultad de Matem\'aticas, Pontificia Universidad Cat\'olica de Chile,
Casilla 306, Correo 22, Santiago, Chile}
\email{diejen@mat.puc.cl, eemsiz@mat.puc.cl}

%\subjclass[2000]{Primary: 43A90; Secondary: 20C08, 33D52}
%\keywords{$q$-bosons, Hall-Littlewood functions, }

\thanks{Work was supported in part by the {\em Fondo Nacional de Desarrollo
Cient\'{\i}fico y Tecnol\'ogico (FONDECYT)} Grants \# 1130226 and  \# 11100315,
and by the {\em Anillo ACT56 `Reticulados y Simetr\'{\i}as'}
financed by the  {\em Comisi\'on Nacional de Investigaci\'on
Cient\'{\i}fica y Tecnol\'ogica (CONICYT)}}

\date{May 2013}

\begin{abstract}
Upon introducing a one-parameter quadratic deformation of the $q$-boson algebra and a diagonal perturbation at the end point, we arrive at a semi-infinite $q$-boson system with a two-parameter boundary interaction.
The eigenfunctions are shown to be given by Macdonald's hyperoctahedral Hall-Littlewood functions of type $BC$. It follows that the $n$-particle
spectrum is bounded and absolutely continuous and that the corresponding scattering matrix factorizes as a product of two-particle bulk and one-particle boundary scattering matrices.
\end{abstract}

\maketitle

\section{Introduction}\label{sec1}
The $q$-boson system \cite{bog-ize-kit:correlation} is a lattice discretization of the one-dimensional quantum nonlinear Schr\"odinger equation \cite{gau:fonction,gut:quantum,kor-bog-ize:quantum,mat:many-body,sut:beautiful} built of particle creation and annihilation operators representing the $q$-oscillator algebra
 \cite[\text{Ch.}~5]{kli-sch:quantum}. Its $n$-particle eigenfunctions are given by Hall-Littlewood functions \cite{tsi:quantum,kor:cylindric,die-ems:diagonalization}.
In the present paper we study a system of $q$-bosons on the semi-infinite lattice with boundary interactions, in the spirit of previous works concerned with the quantum nonlinear Schr\"odinger equation on the half-line
\cite{gau:boundary,gat-lig-min:nonlinear,hal-lan:exact,cau-cra:exact,tra-wid:bose}.

Specifically, by introducing at the end point creation and annihilation operators representing a quadratic deformation of the $q$-oscillator algebra together with a diagonal perturbation, we arrive at the hamiltonian of a $q$-boson system on the semi-infinite integer lattice endowed with a two-parameter boundary interaction. By means of an explicit formula for the action of the hamiltonian in the $n$-particle subspace, it is deduced that the Bethe Ansatz eigenfunctions are given by Macdonald's three-parameter Hall-Littlewood functions with hyperoctahedral symmetry associated with the $BC$-type root system \cite[\S10]{mac:orthogonal}.

It follows that the $q$-boson system fits within a large class of discrete quantum models with bounded absolutely continous spectrum for which the scattering behaviour was determined in great detail by means of stationary phase techniques \cite{rui:factorized,die:scattering}.
In particular, the $n$-particle scattering matrix is seen to factorize as a product of explicitly computed two-particle bulk and one-particle boundary scattering matrices.

\section{Semi-infinite $q$-boson system}\label{sec2}
Let
\begin{equation}\label{AFock}
\mathcal{F}:=\bigoplus_{n\in\mathbb{N}} \mathcal{F}(\Lambda_n)
\end{equation}
 denote the algebraic Fock space consisting of finite linear combinations of $f_n\in\mathcal{F}(\Lambda_n)$, $n\in\mathbb{N}:=\{0,1,2,\ldots \}$,
where
$\mathcal{F}(\Lambda_n)$ stands for the space of functions  $f:\Lambda_n\to\mathbb{C}$
on the set of partitions of length at most $n$:
\begin{equation}\label{dominant}
\Lambda_n:=\{\lambda=(\lambda_1,\dots, \lambda_n)\in\mathbb{N}^n \mid \lambda_1\geq \lambda_2 \geq \cdots \geq \lambda_n\} ,
\end{equation}
with the additional convention that $\Lambda_0:=\{ 0\}$ and $\mathcal{F}(\Lambda_0):=\mathbb{C}$.
For $l\in\mathbb{N}$, we introduce the following actions on $f\in\mathcal{F}(\Lambda_n)\subset \mathcal{F}$:
\begin{equation*}
(\beta_l f)(\lambda):=
  f(\beta_l^*\lambda)
 \qquad\qquad\  (\lambda\in \Lambda_{n-1})
\end{equation*}
if $n>0$ and $\beta_l f:=0$ if $n=0$,
\begin{equation*}
    \begin{split}
(\beta^*_l f)(\lambda)&:=
\begin{cases}
[m_l(\lambda)](1-c\delta_{l}q^{m_0(\lambda)-1})f(\beta_l\lambda)&\text{if}\  m_l(\lambda)>0 \\
0&\text{otherwise}
\end{cases}
 \quad ( \lambda\in \Lambda_{n+1}),\\
(q^{N_l+k} f)(\lambda)&:=q^{m_l(\lambda)+k}f(\lambda) \qquad\quad (\lambda\in \Lambda_n),
    \end{split}
\end{equation*}
with $q,c\in\mathbb{R}$ such that $|q|\neq 0,1$ and $k\in\mathbb{Z}$. Here
$$
\delta_l:= \begin{cases} 1 & \text{for}\ l=0,\\
0 & \text{otherwise}
\end{cases},
\qquad
[m]:=\frac{1-q^{m}}{1-q}=\begin{cases} 0&\text{for}\ m=0\\
1+q+\cdots +q^{m-1}&\text{for}\ m>0
\end{cases},$$
and the multiplicity $m_l(\lambda)$ counts the number of parts
$\lambda_j$, $1\leq j\leq n$ of size $\lambda_j=l$ (so
$m_0(\lambda)$, $\lambda\in\Lambda_n$ is equal to $n$ minus the number of nonzero parts), while
$\beta^*_l\lambda\in\Lambda_{n+1}$ and $\beta_l\lambda\in \Lambda_{n-1}$ stand for the partitions obtained from $\lambda\in\Lambda_n$ by inserting/deleting a part of size $l$, respectively (where it is assumed in the latter situation  that $m_l(\lambda)>0$). It is clear from these definitions that $\beta_l$, $\beta_l^*$ and
$q^{N_l+k}$ map $\mathcal{F}(\Lambda_n)$ into $ \mathcal{F}(\Lambda_{n-1})$,  $ \mathcal{F}(\Lambda_{n+1})$ and  $ \mathcal{F}(\Lambda_{n})$, respectively (with the convention that  $ \mathcal{F}(\Lambda_{-1})$ is the null space).

The operators in question represent a quadratic deformation of the $q$-boson field algebra
at the boundary site $l=0$ parametrized by the constant $c$:
\begin{subequations}
\begin{align}
\beta_l q^{N_l} = q^{N_l+1}\beta_l,\ & \
\beta_l^*q^{N_l} =q^{N_l-1} \beta_l^*  ,\nonumber \\
 \beta_l\beta_l^*=[N_l+1](1-c\delta_{l}q^{N_0}),\ & \ [\beta_l,\beta_l^*]_q =1-c\delta_{l}q^{2N_0}
\end{align}
and preserving the ultralocality:
 \begin{equation}
[\beta_l,\beta_k]=[\beta^*_l,\beta^*_k]=[N_l,N_k]=[N_l,\beta_k]=[N_l,\beta^*_k]=[\beta_l,\beta^*_k]= 0
\end{equation}
\end{subequations}
for $l\neq k$ (where $[A,B]:=AB-BA$, $[A,B]_q:=AB-qBA$, and $[N_l+r]:=(1-q^{N_l+r})/(1-q)$).

When interpreting the characteristic function $|\lambda\rangle\in\mathcal{F}(\Lambda_n)$ supported on $\lambda\in\Lambda_n$ as a state representing a configuration of $n$ particles on $\mathbb{N}$ such that  $m_l(\lambda)$ particles are siting on the site $l\in\mathbb{N}$, it is clear that the operators $\beta_l$ and $\beta_l^*$ act as particle  annihilation and creation operators:
\begin{equation*}
\beta_l |\lambda\rangle = \begin{cases}
|\beta_l\lambda\rangle &\text{if}\ m_l(\lambda)>0 \\
0&\text{otherwise}
\end{cases},\quad
\beta_l^* |\lambda\rangle =[m_l(\lambda)+1](1-c\delta_{l}q^{m_0(\lambda)}) | \beta_l^*\lambda\rangle,
\end{equation*}
while $q^{N_l}$ counts the number of particles at the site $l$ (as a power of $q$):
$$q^{N_l}|\lambda\rangle=q^{m_l(\lambda)}|\lambda\rangle .$$

The dynamics of our  $q$-boson system is governed by a hamiltonian built of left and right hopping operators together with a diagonal boundary term:
\begin{equation}\label{qbH}
\text{H}_q =a[N_0] +\sum_{l\in \mathbb{N}} (\beta_{l+1}\beta^*_l+\beta_{l+1}^*\beta_l ),
\end{equation}
$a\in\mathbb{R}$.
This hamiltonian constitutes a well-defined operator on $\mathcal{F}$ \eqref{AFock} as for any $f\in\mathcal{F}(\Lambda_n)$ and $\lambda\in\Lambda_n$ the
infinite sum $(\text{H}_q f)(\lambda)$ contains only a finite number of nonvanishing terms.

\section{The $n$-Particle hamiltonian and its eigenfunctions}\label{sec3}
By construction $\text{H}_q$ \eqref{qbH}
preserves the $n$-particle subspace $\mathcal{F}(\Lambda_n)$.
The following proposition describes the action of the hamiltonian in this subspace explicitly.

\begin{proposition}[$n$-Particle hamiltonian]\label{action-Hq:prp} For any
 $f\in \mathcal{F}(\Lambda_n)$ and $\lambda\in\Lambda_n$, one has that
\begin{align*}
&(\text{H}_q f)(\lambda)
=  a[m_0(\lambda)] f(\lambda)\ +\\
&\sum_{\substack{1\leq j \leq n\\ \lambda+e_j\in\Lambda_n}} (1-c\delta_{\lambda_j}q^{m_0(\lambda)-1})[m_{\lambda_j}(\lambda)] f(\lambda+e_j)
+\sum_{\substack{1\leq j \leq n\\ \lambda-e_j\in\Lambda_n}} [m_{\lambda_j}(\lambda)]  f(\lambda-e_j) ,
\end{align*}
where $e_1,\ldots ,e_n$ refer to the unit vectors comprising the standard basis of $\mathbb{Z}^n$.
\end{proposition}

\begin{proof}
It is clear from the definitions that
$([N_0]f)(\lambda)=[m_0(\lambda)]f(\lambda)$, and that for any $l\in\mathbb{N}$:
$$
(\beta_{l+1} \beta_l^* f)(\lambda)=
\begin{cases}
[m_l(\lambda)](1-c\delta_lq^{m_0(\lambda)-1})f(\beta_{l+1}^* \beta_l\lambda) &\text{if}\  m_l(\lambda)>0 ,\\
0&\text{otherwise},
\end{cases}$$
where $\beta_{l+1}^* \beta_l\lambda=\lambda+e_j$ with $j=\min\{k\mid \lambda_k=l\}$ (so $l=\lambda_j$),
and
$$(\beta_{l+1}^*\beta_l f)(\lambda)=
\begin{cases}
[m_{l+1}(\lambda)]f(\beta_{l+1}\beta_l^*\lambda)&\text{if}\  m_{l+1}(\lambda)>0 ,\\
0&\text{otherwise},
\end{cases}$$
where $\beta_{l+1}\beta_l^*\lambda=\lambda-e_j$ with $j=\max\{k\mid \lambda_k=l+1\}$ (so $l=\lambda_j-1$).
\end{proof}

The $n$-particle hamiltonian has Bethe Ansatz eigenfunctions given by the following plane wave expansion
\begin{subequations}
\begin{equation}\label{HLf}
\phi_\xi (\lambda) :=
\sum_{\substack{\sigma\in S_n\\ \epsilon\in \{\pm 1\}^n}}
C(\epsilon \xi_{\sigma}) e^{i \langle \lambda, \epsilon \xi_{\sigma}\rangle} ,
\end{equation}
with expansion coefficients of the form
\begin{eqnarray}\label{Cf}
\lefteqn{C(\xi) :=
\prod_{1\leq j\leq n} \frac{1-ae^{-i\xi_j}+ce^{-2i\xi_j}}{1-e^{-2i\xi_j}}} && \\
&& \times \prod_{1\leq j<k \leq n} \Bigl(\frac{1-q e^{-i(\xi_{j}-\xi_k)}}{1-e^{-i(\xi_{j}-\xi_k)}}\Bigr)\Bigl(  \frac{1-q e^{-i(\xi_{j}+\xi_k)}}{1-e^{-i(\xi_{j}+\xi_k)}} \Bigr) . \nonumber
\end{eqnarray}
\end{subequations}
Here $\langle\cdot,\cdot\rangle$ denotes the standard inner product on $\mathbb{R}^n$,
$\epsilon\xi_\sigma:=(\epsilon_1\xi_{\sigma_1},\epsilon_2\xi_{\sigma_2},\ldots ,\epsilon_n\xi_{\sigma_n})$,
and the summation is meant over all permutations $\sigma$ in the symmetric group $S_n$ and all sign configurations
$\epsilon=(\epsilon_1,\ldots,\epsilon_n)\in \{ 1,-1\}^n$.
Viewed as a function of the spectral parameter $\xi=(\xi_1,\ldots ,\xi_n)$ in the fundamental alcove
\begin{equation}\label{alcove}
 A:=\{ (\xi_1,\xi_2,\ldots,\xi_n)\in\mathbb{R}^n\mid \pi>\xi_1>\xi_2>\cdots >\xi_n>0\} ,
\end{equation}
the expression $\phi_\xi(\lambda)$, $\lambda\in\Lambda_n$ amounts to Macdonald's three-parameter Hall-Littlewood polynomial with hyperoctahedral symmetry associated with the root system $BC_n$ \cite[\S10]{mac:orthogonal}.

\begin{proposition}[Bethe Ansatz eigenfunctions]\label{BAE:prp}
The $n$-particle Bethe Ansatz wave function $\phi_\xi$, $\xi\in A$ solves the eigenvalue equation
\begin{equation}\label{ev:eq}
\text{H}_q \phi_\xi = E_n(\xi) \phi_\xi,\qquad E_n(\xi):=2\sum_{j=1}^n \cos (\xi_j).
\end{equation}
\end{proposition}

\begin{proof}
It follows from Proposition \ref{action-Hq:prp} that
the stated eigenvalue equation boils down to the following identity
\begin{align*}
 a[m_0(\lambda)] \phi_\xi(\lambda)\  &+\
\sum_{\substack{1\leq j \leq n\\ \lambda+e_j\in\Lambda_n}} (1-c\delta_{\lambda_j}q^{m_0(\lambda)-1})[m_{\lambda_j}(\lambda)] \phi_\xi (\lambda+e_j) \\
&+\sum_{\substack{1\leq j \leq n\\ \lambda-e_j\in\Lambda_n}} [m_{\lambda_j}(\lambda)] \phi_\xi(\lambda-e_j) =
2\phi_\xi(\lambda ) \sum_{j=1}^n \cos (\xi_j) ,
\end{align*}
which is in turn equivalent to the Pieri formula  for the hyperoctahedral Hall-Littlewood function in Eq. \eqref{pieri2} of Appendix \ref{appA}.
\end{proof}

\section{Diagonalization}\label{sec4}
From now on it will be assumed unless stated otherwise that $0<|q|<1$ and that the boundary parameters
$a$ and $c$  are chosen such that the roots $r_1$, $r_2$ of the quadratic polynomial $r^2-ar+c$ belong to the interval $(-1,1)$:
\begin{equation}\label{od}
a=r_1+r_2\ \text{and}\  c=r_1r_2\ \text{with}\ r_1,r_2\in (-1,1) .
\end{equation}

Let $ L^2(A,\Delta\text{d}\xi)$ be the Hilbert space of functions $\hat{f}:A\to\mathbb{C}$ characterized by the inner product
\begin{equation}\label{ip}
\langle \hat{f},\hat{g}\rangle_\Delta=\frac{1}{(2\pi)^n}\int_A \hat{f}(\xi)\overline{\hat{g}(\xi)}\Delta(\xi)\text{d}\xi ,
\quad \text{where}\quad \Delta (\xi):= \frac{1}{|C(\xi)|^2}
\end{equation}
with $C(\xi)$ given by Eq. \eqref{Cf}.
It is well-known that for the parameter regime in question
Macdonald's hyperoctahedral Hall-Littlewood functions form an orthogonal basis
of $ L^2(A,\Delta\text{d}\xi)$ \cite[\S 10]{mac:orthogonal}:
\begin{subequations}
\begin{equation}\label{orthogonality}
\langle \phi(\lambda) ,\phi (\mu) \rangle_\Delta =\begin{cases}
\mathcal{N}(\lambda)&\text{if}\ \lambda =\mu ,\\
0&\text{otherwise},
\end{cases}
\end{equation}
where
\begin{equation}\label{norm}
\mathcal{N}(\lambda):= (c;q)_{m_0(\lambda)} \prod_{\ell\in\mathbb{N}} [m_\ell (\lambda)]!
\end{equation}
\end{subequations}
with $(c;q)_m:=(1-c)(1-cq)\cdots (1-cq^{m-1})$ (and the convention that $(c;q)_0:=1$) and
$[m]!:=(q;q)_m/(q;q)_1^m=[m][m-1]\cdots [2][1]$.
By combining the orthogonality in Eqs. \eqref{orthogonality}, \eqref{norm} with Proposition \ref{BAE:prp},
the spectral decomposition of $\text{H}_{q}$ in the $n$-particle Hilbert space
$\ell^2(\Lambda_n,\mathcal{N}^{-1})\subset\mathcal{F}(\Lambda_n)$ characterized by the inner product
 \begin{equation}
 \langle f,g\rangle_n:=\sum_{\lambda\in\Lambda_n} f(\lambda) \overline{g(\lambda)} \mathcal{N}^{-1}(\lambda)
 \end{equation}
becomes immediate.

 \begin{theorem}[Diagonalization]\label{diagonal:thm}
 For $0< |q|<1$ and values of the boundary parameters $a$ and $c$ in the orthogonality domain \eqref{od},
 the $q$-boson Hamiltonian $\text{H}_{q}$ \eqref{qbH} restricts to a bounded self-adjoint operator in $\ell^2(\Lambda_n,\mathcal{N}^{-1})$ with purely absolutely continuous spectrum. More specifically, its
 spectral decomposition reads explicitly
 \begin{equation}\label{s-d}
\text{H}_{q}=\boldsymbol{F_q}^{-1}  \circ \hat{{E}} \circ\boldsymbol{F_q},
 \end{equation}
 where $\boldsymbol{F_q}:\ell^2(\Lambda_n,\mathcal{N}^{-1})\to L^2(A,\Delta\text{d}\xi)$ denotes the unitary Fourier transform associated with the hyperoctahedral Macdonald-Hall-Littlewood basis:
\begin{subequations}
\begin{equation}\label{ft1}
(\boldsymbol{F_q}f)(\xi):= \langle f,\phi_\xi \rangle_n=\sum_{\lambda\in\Lambda_n}f(\lambda)
\overline{\phi_\xi (\lambda)}\mathcal{N}^{-1}(\lambda)
\end{equation}
($f\in \ell^2(\Lambda_n,\mathcal{N}^{-1})$) with the inversion formula given by
\begin{equation}\label{ft2}
(\boldsymbol{F_q}^{-1}\hat{f})(\lambda) = \langle \hat{f},\overline{\phi ( \lambda)}\rangle_\Delta=
\frac{1}{(2\pi)^n}\int_A \hat{f}(\xi) \phi_\xi (\lambda)\Delta(\xi)\text{d}\xi
\end{equation}
\end{subequations}
$(\hat{f}\in L^2(A,\Delta\text{d}\xi))$,
 and
 $(\hat{{E}}\hat{f})(\xi):=E_n(\xi) \hat{f}(\xi)$ stands for the bounded real multiplication operator in
$L^2(A,\Delta\text{d}\xi)$ associated with the $n$-particle eigenvalue $E_n(\xi)$ \eqref{ev:eq}.
\end{theorem}

In the Fock space
$
\mathcal{H}:=\bigoplus_{n\geq 0} \ell^2(\Lambda_n,\mathcal{N}^{-1})
$,
built of all linear combinations $\sum_{n\geq 0} c_n f_n$ with $c_n\in\mathbb{C}$ and $f_n\in \ell^2(\Lambda_n,\mathcal{N}^{-1})$ such that $\sum_{n\geq 0} |c_n|^2 \langle f_n,f_n\rangle_n<\infty$,
the $q$-boson hamiltonian $\text{H}_q$ \eqref{qbH} constitutes an unbounded
operator that is essentially self-adjoint on the dense domain $\mathcal{D}:=\mathcal{F}\cap\mathcal{H}$ (because
for $z\in\mathbb{C}\setminus\mathbb{R}$ the range $(\text{H}_{q}-z)\mathcal{D}$ is dense in $\mathcal{H}$
and $\lim_{n\to\infty} \sup_{\xi\in A} |E_n(\xi)|=\infty$).
The representation of the deformed $q$-boson field algebra in Section \ref{sec2} on the other hand gives rise to a bounded representation
on $\mathcal{H}$:
\begin{align*}
\langle \beta_l f,\beta_l f\rangle_{n-1}&\leq  \frac{1+|c|\delta_{l}}{1-q} \langle  f,f\rangle_{n}  ,\nonumber\\
\langle \beta_l^* f,\beta_l^*f\rangle_{n+1}&\leq  \frac{1+|c|\delta_{l}}{1-q}  \langle  f,f\rangle_{n},\\
\langle q^{N_l} f,q^{N_l}f\rangle_{n}&\leq  \langle  f,f\rangle_{n} , \nonumber
\end{align*}
preserving the $*$-structure:
\begin{equation*}
\langle \beta_l^*f,g\rangle_{n+1}= \langle f, \beta_l g\rangle_{n}\quad\text{and}\quad
\langle q^{N_l}f,g\rangle_{n}= \langle f, q^{N_l} g\rangle_{n} .
\end{equation*}

\begin{remark}
Upon rescaling the lattice $\Lambda_n$ \eqref{dominant} and performing an appropriate continuum limit  \cite[Sec.~5]{die:plancherel}, Macdonald's hyperoctahedral Hall-Littlewood functions tend to the eigenfunctions of the quantum nonlinear Schr\"odinger equation on the half-line with a boundary interaction \cite{gau:boundary,gat-lig-min:nonlinear,hal-lan:exact,cau-cra:exact,tra-wid:bose}. In particular, it follows from
 \cite[Sec.~5.3]{die:plancherel} that for $a=0$ (which corresponds to a reduction from type $BC$ to type $C$ root systems) a renormalized version of the $q$-boson hamiltonian $\text{H}_q$ \eqref{qbH}
then converges in the $n$-particle subspace in the strong resolvent sense to a hamiltonian that can be written formally as:
 \begin{equation*}
 -\sum_{j=1}^n \frac{\partial^2}{\partial x_j^2}+
g \sum_{1\leq j< k\leq n} \bigl( \delta (x_j-x_k)+\delta (x_j+x_k)\bigr)
 +g_0 \sum_{1\leq j\leq n}\delta(x_j)
 \end{equation*}
with $g, g_0> 0$ (where $\delta (\cdot )$ stands for the `delta potential').
 \end{remark}

\section{Factorized scattering}
The similarity transformation
\begin{equation}\label{Hq-lebesgue}
H :=\mathcal{N}^{-1/2}\,\text{H}_{q}\,\mathcal{N}^{1/2}
\end{equation}
turns
the $n$-particle $q$-boson hamiltonian in Proposition \ref{action-Hq:prp} into a self-adjoint operator in $\ell^2(\Lambda_n)$ diagonalized by the normalized wave function
\begin{subequations}
\begin{equation}
\begin{split}
 \Psi_\xi(\lambda)&:= e^{\frac{\pi i}{2}n^2}  |C (\xi)|^{-1} \mathcal{N}(\lambda )^{-1/2} \phi_\xi (\lambda ) \\
 &=\mathcal{N}(\lambda )^{-1/2}\sum_{\substack{\sigma \in S_n\\ \epsilon\in \{ \pm 1\}^n}}\text{sign} (\epsilon \sigma )
\hat{\mathcal S}(\epsilon \xi_\sigma )^{1/2} e^{i\langle \rho+\lambda, \epsilon\xi_\sigma\rangle} ,
 \end{split}
\end{equation}
with $\xi\in A$ \eqref{alcove}, $\text{sign} (\epsilon \sigma ):=\epsilon_1\cdots\epsilon_n\text{sign} ( \sigma )$,  $\rho:=(n-1,n-2,\ldots ,2,1,0)$,  and
\begin{equation}
\hat{ {\mathcal S}} (\xi)
 := \prod_{1\leq j<k\leq n} s(\xi_j-\xi_k)s(\xi_j+\xi_k)\prod_{1\leq j\leq n} s_0(\xi_j) ,
\end{equation}
where
\begin{align}
s(x):=\frac{1-qe^{-ix}}{1-qe^{ix}}\quad&\text{with}\quad      s(x)^{1/2}=\frac{1-qe^{-ix}}{|1-q e^{ix}|}
\intertext{and}
s_0(x):=\frac{1-ae^{-ix}+ce^{-2ix}}{1-ae^{ix}+ce^{2ix}}
\quad&\text{with}\quad
 s_0(x)^{1/2}=\frac{1-ae^{-ix}+ce^{-2ix}}{|1-ae^{ix}+ce^{2ix}|}.
\end{align}
\end{subequations}
Specifically, one has that
$
H=\boldsymbol{F}^{-1}  \circ \hat{{E}} \circ\boldsymbol{F}
$
where
$\boldsymbol{F}:\ell^2(\Lambda_n)\to  L^2(A,\text{d}\xi)$ denotes the unitary Fourier transformation
determined by the kernel $\Psi_\xi(\lambda)$ (and
 $\hat{E}$ is now interpreted as a bounded multiplication operator in $  L^2(A,\text{d}\xi)$). For $q,a,c\to 0$ the
$n$-particle $q$-boson hamiltonian  $H$ \eqref{Hq-lebesgue} simplifies to a hamiltonian modeling
 impenetrable bosons on $\mathbb{N}$:
   \begin{equation*}
(H_0 f)(\lambda)
=
\sum_{\substack{1\leq j \leq n\\ \lambda+e_j\in\Lambda_n}} f(\lambda+e_j)
+\sum_{\substack{1\leq j \leq n\\ \lambda-e_j\in\Lambda_n}}   f(\lambda-e_j)
\end{equation*}
($f\in\ell^2(\Lambda_n)$), which is diagonalized by the conventional Fourier transform $\boldsymbol{F_0}:\ell^2(\Lambda_n)\to  L^2(A,\text{d}\xi)$ obtained
from $\boldsymbol{F}$ by setting $\hat{\mathcal{S}}(\xi)\equiv 1$, $\mathcal{N}(\lambda)\equiv 1$.

As a very special case of the results in \cite[\text{Sec.}~4]{die:scattering}, it now follows that
the wave- and scattering operators  comparing the $q$-boson dynamics
\begin{equation}\label{dynamics}
(e^{itH }f)(\lambda)=
\frac{1}{(2\pi)^n}\int_A e^{itE_n (\xi )}\hat{f}(\xi) \Psi_\xi (\lambda)\text{d}\xi\qquad
\hat{f}= \boldsymbol{F} f
\end{equation}
 with the corresponding impenetrable boson dynamics generated by $H_0$ are governed by a unitary $S$-matrix
 $ \hat{\mathcal S} : L^2(A,\text{d}\xi)\to  L^2(A,\text{d}\xi)$
 of the form
 \begin{equation}
 ( \hat{\mathcal S}\hat{f})(\xi):=  \hat{\mathcal S}(\epsilon_\xi \xi_{\sigma_\xi})\hat{f}(\xi)\qquad
 (\hat{f}\in C_0(A_r).
  \end{equation}
 Here $C_0(A_r)$ denotes the dense subspace of $L^2(A,\text{d}\xi)$ consisting of smooth test functions with compact support in
 the open dense subset $A_r\subset A$ for which the components of
 $\nabla E_n(\xi)=(-2\sin(\xi_1),\ldots,-2\sin(\xi_n))$ do not vanish and are all distinct in absolute value, and
  the sign-configuration $\epsilon_\xi$ and the permutation $\sigma_\xi$ are such that the components
  of  $\nabla E_n(\epsilon_\xi \xi_{\sigma_\xi})$ are all positive and ordered from large to small.
 Specifically, by comparing the large-time asymptotics of oscillatory integrals of the form in Eq. \eqref{dynamics} for the dynamics generated by $H$ and $H_0$ one concludes that \cite[\text{Thm.}~4.2 \text{and} \text{Cor.}~4.3]{die:scattering}:

\begin{theorem}[Wave and scattering operators]\label{scattering:thm}
  The operator limits
\begin{subequations}
\begin{equation}
\Omega^{\pm} :=s-\lim_{t\to \pm \infty}  e^{i t  H}e^{-it H_{0}}
\end{equation}
converge in the strong $\ell^2(\Lambda_n)$-norm topology and the corresponding wave operators $\Omega^\pm_r$ are given by unitary operators in $\ell^2(\Lambda_n)$ of the form
\begin{equation}
\Omega_r^\pm = \boldsymbol{F}^{-1} \circ \hat{\mathcal S}^{\mp 1/2}  \circ \boldsymbol{F_0}.
\end{equation}
Hence, the scattering operator comparing the dynamics of $H$  and $H_{0}$ is given by the unitary operator
\begin{equation}
\mathcal{S}:=(\Omega_r^+)^{-1} \Omega_r^- =  \boldsymbol{F_0}^{-1}  \circ \hat{\mathcal S} \circ  \boldsymbol{F_0} .
\end{equation}
\end{subequations}
\end{theorem}

\appendix

\section{Pieri formula for Macdonald's hyperoctahedral Hall-Littlewood function}\label{appA}
Let $x:=(x_1,\ldots ,x_n)=(e^{i\xi_1},\ldots,e^{i\xi_n})$
and
$\tau:=(\tau_1,\ldots,\tau_n)$, where
$\tau_j=rq^{n-j}$ ($j=1,\ldots ,n$) with $r=\frac{a}{2}+\sqrt{(\frac{a}{2})^2-c}$ (cf. Eq. \eqref{od}).
Upon setting
\begin{equation}
P_\lambda (x):= \frac{\tau_1^{\lambda_1}\cdots \tau_n^{\lambda_n}}{\mathcal{N}(0)}\phi_\xi(\lambda)\qquad
(\lambda\in\Lambda_n),
\end{equation}
where $\mathcal{N}(0)$ is given by Eq. \eqref{norm} with $\lambda=0$, the
hyperoctahedral Hall-Littlewood function is renormalized to have
unital principal specialization values: $P_\lambda(\tau)=1$ ($\forall\lambda\in\Lambda_n)$ \cite[\S 10]{mac:orthogonal}.
With this normalization, the following Pieri formula holds:
\begin{eqnarray}\label{pieri}
\lefteqn{P_\lambda (x) \sum_{j=1}^n (x_j+x_j^{-1}-\tau_j-\tau_j^{-1})=} &&\\
&&\sum_{\substack{1\leq j\leq n\\ \lambda +e_j\in\Lambda_n}}
V_j^+(\lambda) \left( P_{\lambda +e_j}(x)-P_\lambda(x) \right)
+\sum_{\substack{1\leq j\leq n\\ \lambda -e_j\in\Lambda_n}}
V_j^-(\lambda) \left( P_{\lambda -e_j}(x)-P_\lambda(x) \right)  ,\nonumber
\end{eqnarray}
where
\begin{align*}
V_j^+(\lambda) &= \tau_j^{-1}\Bigl(\frac{1-c^2\delta_{\lambda_j}q^{2(n-j)}}{1+c\delta_{\lambda_j}q^{2(n-j)}}\Bigr)
\prod_{\substack{j<k\leq n\\ \lambda_k=\lambda_j}}
\Bigl(\frac{1-q^{1+k-j}}{1-q^{k-j}}\Bigr)\Bigl( \frac{1+c\delta_{\lambda_j}q^{1+2n-k-j}}{1+c\delta_{\lambda_j}q^{2n-k-j}}\Bigr),\\
V_j^-(\lambda) &= \tau_j \prod_{\substack{1\leq k<j\\ \lambda_k=\lambda_j}}
\Bigl(\frac{1-q^{1+j-k}}{1-q^{j-k}}\Bigr) .
\end{align*}

The formula in question is readily obtained through degeneration from an analogous Pieri formula for a $BC_n$-type Macdonald function that arises as a special case of the Pieri formulas in
\cite[\text{Sec}.~6.1]{die:properties}. Specifically, by substituting $t_2=q^{1/2}$, $t_3=-q^{1/2}$ (which amounts to a reduction from the Macdonald-Koornwinder function to the $BC_n$-type Macdonald function) in the Pieri formula of \cite[\text{Eqs}.~(6.4), (6.5)]{die:properties} with coefficients taken from \cite[\text{Eqs}.~(6.12), (6.13)]{die:properties}, the relation in Eq. \eqref{pieri} is retrieved for
$q\to 0$ (which corresponds to a transition from Macdonald type functions to Hall-Littlewood type functions).  Notice in this connection that the
parameters $q$, $a$, $c$ (and $r$) of the present paper are related to the parameters
$t$, $t_0$, $t_1$  of Ref. \cite{die:properties} via  $q=t$, $a=t_0+t_1$,
$c=t_0t_1$ (and $r=t_0$).

Since
\begin{equation*}
V_j^+(\lambda)=\tau_j^{-1}(1-c\delta_{\lambda_j}q^{m_0(\lambda)-1})[m_{\lambda_j}(\lambda)]  ,
\qquad  V_j^-(\lambda)= \tau_j [m_{\lambda_j}(\lambda)] ,
\end{equation*}
and
$$
\sum_{j=1}^n( \tau_j+\tau_j^{-1})
-\sum_{\substack{1\leq j \leq n\\ \lambda-e_j\in\Lambda_n}} \tau_j [m_{\lambda_j}(\lambda)]
-\sum_{\substack{1\leq j \leq n\\ \lambda+e_j\in\Lambda_n}}
 \tau_j^{-1} [m_{\lambda_j}(\lambda)]
 =r[m_0(\lambda)]  ,
$$
the Pieri formula \eqref{pieri} can be condensed into the more compact form
\begin{align}\label{pieri2}
P_\lambda (x) \sum_{j=1}^n (x_j+x_j^{-1})&=a[m_0(\lambda)] +\sum_{\substack{1\leq j\leq n\\ \lambda -e_j\in\Lambda_n}}
\tau_j [m_{\lambda_j}(\lambda)] P_{\lambda -e_j}(x) \\
&
+\sum_{\substack{1\leq j\leq n\\ \lambda +e_j\in\Lambda_n}}
 \tau_j^{-1} (1-c\delta_{\lambda_j}q^{m_0(\lambda)-1})[m_{\lambda_j}(\lambda)] P_{\lambda +e_j}(x)
 .\nonumber
\end{align}

\bibliographystyle{amsplain}

\end{document}